\documentclass{article}  
\usepackage{ltexpprt}

\usepackage{amsmath,amsfonts,amssymb,graphicx,graphics}
\usepackage{mathrsfs}

\usepackage{fge}
\usepackage{xspace}

\graphicspath{{img/}} 

\newcommand{\sepP}{\varphi}
\newcommand{\sepQ}{\gamma}
\newcommand{\pointP}{p^*}
\newcommand{\pointQ}{q^*}

\newcommand{\total}{n}

\newcommand{\dualshell}[1]{\ensuremath{\mathbb S({#1})}}
\newcommand{\shell}[1]{\ensuremath{S({#1})}}
\newcommand{\ch}[1]{\ensuremath{\textsc{ch}(#1)}}

\newcommand{\polar}[1]{\ensuremath{\rho({#1})}}

\newcommand{\zero}{\ensuremath{\fgestruckzero}}

\newcommand{\hier}{BDK-hierarchy\xspace}
\newcommand{\hiers}{BDK-hierarchies\xspace}
\newcommand{\bounded}{$c$-bounded\xspace}

\newcommand{\phinf}[1]{\ensuremath{\textsc{ph}_{_\infty}}[#1]}
\newcommand{\phzero}[1]{\ensuremath{\textsc{ph}_{_\zero}}[#1]}

\newcommand{\cone}[2]{\ensuremath{\kappa_{#1}(#2)}}
\newcommand{\neighbors}[2]{\ensuremath{N_{#1}(#2)}}

\newcommand{\polarzero}[1]{\ensuremath{\rho_{_\zero}(#1)}}
\newcommand{\polarinf}[1]{\ensuremath{\rho_{_\infty}(#1)}}

\newcommand{\planezero}[1]{\ensuremath{#1_{_\zero}}}
\newcommand{\planeinf}[1]{\ensuremath{#1_{_\infty}}}

\newcommand{\planezeropolar}[1]{\ensuremath{\rho_{_\zero}(#1)}}
\newcommand{\planeinfpolar}[1]{\ensuremath{\rho_{_\infty}(#1)}}

\begin{document}

\title{\Large Optimal detection of intersections between convex polyhedra}

\author{Luis Barba\thanks{Department of Computer Science, ETH Z\"urich, Switzerland, \texttt{luis.barba@inf.ethz.ch}} \and Stefan Langerman\thanks{D\'epartement d'Informatique, Universit\'e Libre de Bruxelles, Brussels, Belgium, \tt{slanger@ulb.ac.be}} \thanks{Directeur de recherches du F.R.S.-FNRS.} }

\date{}

\maketitle
\begin{abstract} \small\baselineskip = 10pt
For a polyhedron $P$ in $\mathbb{R}^d$, denote by $|P|$ its combinatorial complexity, i.e., 
the number of faces of all dimensions of the polyhedra.
In this paper, we revisit the classic problem of preprocessing
polyhedra independently so that given two preprocessed polyhedra $P$
and $Q$ in $\mathbb{R}^d$, each translated and rotated, their
intersection can be tested rapidly.

For $d=3$ we show how to perform such a test in 
$O(\log |P| + \log |Q|)$ time
after linear preprocessing time and space. This running time is the best possible and improves upon the last
best known query time of $O(\log|P| \log|Q|)$ by Dobkin and Kirkpatrick (1990).

We then generalize our method to any constant dimension $d$, achieving the same optimal
 $O(\log |P| + \log |Q|)$ query time using a representation of size
 $O(|P|^{\lfloor d/2\rfloor + \varepsilon})$ for any $\varepsilon>0$ arbitrarily small.
This answers an even older question posed by Dobkin and Kirkpatrick 30
years ago.

In addition, we provide an alternative $O(\log |P| + \log |Q|)$
algorithm to test the intersection of two convex polygons $P$ and $Q$ in the plane.
\end{abstract}

\section{Introduction}

Constructing or detecting the intersection between geometric objects 
has been an important subject of study in computational
geometry. It was one of the main questions addressed in
Shamos' seminal paper that lay the grounds of computational
geometry~\cite{shamos1975geometric}, the first application of the
plane sweep technique~\cite{shamos1976geometric}, and is still the topic of several volumes
being published today. 

Extensive research has focused on finding efficient algorithms
for intersection testing or collision detection as this class of
problems has countless applications in motion planning, robotics,
computer graphics, Computer-Aided Design, VLSI design and more.
For information on collision detection refer
to surveys~\cite{jimenez20013d, lin1998collision} and to~Chapter~38 of the
Handbook of Computational Geometry~\cite{CRCHandbook2004}.


The first problem to be addressed is to compute the
intersection of two convex objects. In this paper we focus on convex
polygons and convex polyhedra (or simply polyhedra).
Let $P$ and $Q$ be two polyhedra to be tested for intersection. Let $|P|$ and $|Q|$ denote the combinatorial complexities of $P$ and $Q$, respectively, i.e., the number of faces of all dimensions of the polygon or polyhedra (vertices are 0-dimensional faces while edges are 1-dimensional faces).
Let $\total = |P| + |Q|$ denote the total complexity. 

In the plane, Shamos~\cite{shamos1975geometric} presented an optimal
$\Theta(\total)$-time algorithm to construct the intersection of a
pair of convex polygons. Another linear time algorithm was later presented by
O'Rourke et al.~\cite{O'Rourke1982384}.
In 3D space, Muller and Preparata~\cite{muller1978finding} proposed an $O(\total\log \total)$ time algorithm to test whether two polyhedra in three-dimensional space intersect. Their algorithm has a second phase which computes the intersection of these polyhedra within the same running time using geometric dualization. 
Dobkin and Kirkpatrick~\cite{dobkin1985linear} introduced a hierarchical data structure to represent a polyhedron that allows them to test if two polyhedra intersect in linear time. In a subsequent paper, Chazelle~\cite{Chazelle92anoptimal} presented an optimal linear time algorithm to compute the intersection of two polyhedra in 3D-space.

A natural extension of this problem is to consider the effect of
preprocessing on the complexity of intersection detection problems. 
In this case, significant improvements are possible in the query time. It is worth noting that each object should be preprocessed separately which allows us to work with large families of objects and to introduce new objects without triggering a reconstruction of the whole structure.

Chazelle and Dobkin~\cite{ChazelleD80,Chazelle1987} were the first to formally
define and study this class of problems and provided an algorithm
running in $O(\log|P|+\log|Q|)$ time to test the intersection of two
convex polygons $P$ and $Q$ in the plane. An alternate solution was
given by Dobkin and Kirkpatrick~\cite{dobkin1983fast} with the same
running time. Edelsbrunner~\cite{edelsbrunner1985computing} then used
that algorithm as a preprocessing phase to find the closest pair of
points between two convex polygons, within the same running time.
Dobkin and Souvaine~\cite{dobkin1991detecting} extended these algorithms to test the intersection of two convex planar regions with piecewise curved boundaries of bounded degree in logarithmic time.
These separation algorithms rely on an involved case analysis to solve
the problem. By parameterizing the boundary of $P$ and $Q$, the problem of
determining the closest pair between two polygons can
be seen as finding a minimum of a (discrete) bivariate function.  In an attempt
to simplify these algorithms, Demaine and
Langerman~\cite{unimodal} presented a detailed analysis of what
properties are sufficient in order to be able to compute a minimum of
such a function in logarithmic time.

In Section~\ref{sec:algorithm-plane}, we show an
alternate (and hopefully simpler) algorithm to determine if two convex polygons
$P$ and $Q$ intersect in $O(\log|P|+\log|Q|)$ time.

In all these 2D algorithms, preprocessing is unnecessary if the
polygon is represented by an array with the vertices of the polygon in
sorted order along its boundary.  
In 3D-space (and in higher dimensions) however, the need for
preprocessing is more evident as the traditional DCEL representation
of the polyhedron is not sufficient to perform fast queries.

In this setting, Chazelle and Dobkin~\cite{Chazelle1987} presented a
method to preprocess a 3D polyhedron
and use this structure to test if two preprocessed
polyhedra intersect in $O(\log^3 \total)$ time.
Dobkin and Kirkpatrick~\cite{dobkin1983fast} unified and extended
these results,
showing how to detect if two independently preprocessed
polyhedra intersect in $O(\log^2 \total)$ time.  
Both methods represent a polyhedron $P$ by storing parallel slices of
$P$ through each of its vertices, and thus require $O(|P|^2)$ time,
although space usage could be reduced using persistent data
structures.

In 1990, Dobkin and
Kirkpatrick~\cite{dobkin1990determining} proposed a fast query algorithm that 
uses the linear space hierarchical representation of a polyhedron $P$
defined in their previous article~\cite{dobkin1985linear}. Using this
structure, they show how to determine in $O(\log|P|\log|Q|)$ time if
the polyhedra $P$ and $Q$ intersect. They achieve this by maintaining 
the closest pair between subsets of the polyhedra $P$ and $Q$ as the
algorithms walks down the hierarchical representation. 
While a naive implementation of this algorithm could take time $\Omega(|P|+|Q|)$, 
O'Rourke~\cite[Chapter~7]{o1998computational} describes in detail 
an implementation that avoids this issue and restores the $O(\log|P|\log|Q|)$ bound.
In Section~\ref{section:Bounded degree polyhedron}, we detail the
specifics of this issue, and then we provide a simple modification of
the hierarchical representation that offers an alternative solution.

Whether the intersection of two preprocessed polyhedra $P$ and $Q$ can be tested
in $O(\log|P| + \log|Q|)$ time is an open question that was implicit
in the paper of Chazelle and Dobkin~\cite{ChazelleD80} in STOC'80, and
explicitly posed in 1983 by Dobkin and Kirkpatrick~\cite{dobkin1983fast}. 
More recently, the open problem was listed again in 2004 by David Mount in Chapter~38 of the Handbook
of Computational Geometry~\cite{CRCHandbook2004}. 
Together with this question in 3D-space, Dobkin and
Kirkpatrick~\cite{dobkin1983fast} asked if it is possible to extend
these result to higher dimensions, i.e., to independently preprocess
two polyhedra in $\mathbb{R}^d$ such that their intersection
could be tested in $O(\log \total)$ time. 

\begin{figure*}
\centering
\includegraphics{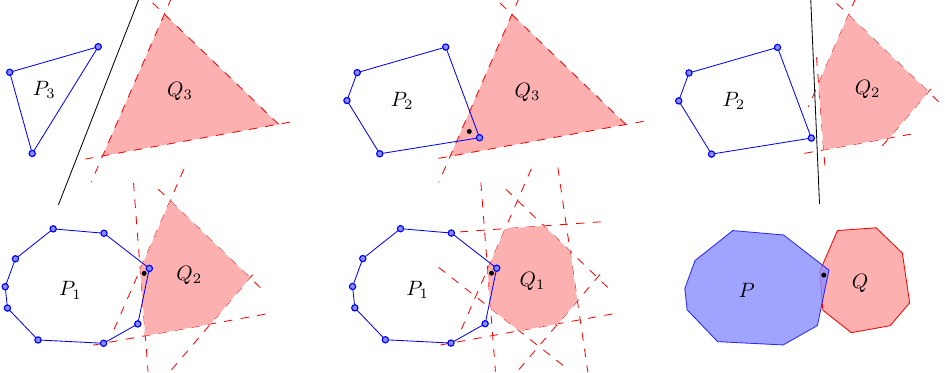}
\caption{\small In each step, the algorithm moves down in either the internal hierarchy of $P$, say $P_1, P_2, P_3$, or the external hierarchy of $Q$, say $Q_1, Q_2, Q_3$. Throughout, the polyhedron in the hierarchy of $P$ grows while the polyhedron in the hierarchy of $Q$ shrinks. A separating (black) line or an intersection (black) point is maintained in each step.}
\label{fig:DualHierarchy}
\end{figure*}

These running times are best possible as, even in the plane, testing if a point intersects a regular $m$-gon $M$ has a lower bound of $\Omega(\log m)$ in the algebraic decision tree model. 

In this paper, we match this lower bound by showing how to independently preprocess polyhedra $P$ and $Q$ in any bounded dimension such that their intersection
can be tested in $O(\log \total)$ time\footnote{In this paper, all algorithms are in the real RAM model of computation.}.
In Section~\ref{section:Bounded degree polyhedron}, we show how to preprocess a polyhedron $P\in \mathbb{R}^3$ in
linear time to obtain a linear space representation. 
In Section~\ref{section:3D algorithm} we provide an algorithm that,
given any translation and rotation of two preprocessed polyhedra $P$
and $Q$ in $\mathbb{R}^3$, tests if they intersect in $O(\log |P| + \log|Q|)$ time. 
In Section~\ref{sec:high-d} we generalize our results to any constant
dimension $d$ and show a representation that allows to test if two
polyhedra $P$ and $Q$ in $\mathbb{R}^d$ (rotated and translated)
intersect in $O(\log|P| + \log|Q|)$ time. The space required by the
representation of a polyhedron $P$ is then $O(|P|^{\lfloor d/2\rfloor + \varepsilon})$ for any small $\varepsilon>0$.
This increase in the space requirements for $d \geq 4$ is not unexpected
as the problem studied here is at least as hard as performing
halfspace emptiness queries for a set of $m$ points in $\mathbb{R}^d$. For this problem, the
best known log-query data structures use roughly $O(m^{\lfloor d/2\rfloor})$ space~\cite{matouvsek1993ray}, and super-linear space lower bounds are known for $d \geq 5$~\cite{erickson2000space}.

\subsection{Outline}
To guide the reader, we give a rough sketch of the algorithm presented in this paper, which is illustrated in Figure~\ref{fig:DualHierarchy}.

We use two types of hierarchical structures of logarithmic depth to represent a polyhedron. 
An \emph{internal hierarchy} is obtained by recursively removing ``large'' sets of the vertices of the polyhedron and taking the convex hull of the remaining vertices.
Since a polyhedron can also be seen as the intersection of halfspaces, 
an \emph{external hierarchy} can be obtained by recursively removing ``large'' sets of halfspaces and taking the intersection of the remaining halfspaces. (A similar structure was introduced by Dobkin et al. to test how ``deeply'' two polyhedra intersect~\cite{dobkin1993computing}). Thus, at the top of these hierarchies we store constant size polyhedra, while at the bottom the full polyhedra are stored.

To test two preprocessed polyhedra $P$ and $Q$ for intersection, 
we use an inner hierarchy for $P$ and an external hierarchy for $Q$. 
Starting at the top, we make our way down by moving one step at the time in either hierarchy.
We move down in the hierarchy of $P$ by adding more vertices (which increases its size), while we move down in the hierarchy of $Q$ by adding halfspace constraints (which decreases its size). 
Thus, in our algorithm one polyhedron grows while the other shrinks, whereas previous approaches grew both polyhedra simultaneously.
Additionally, we maintain either a separating plane or an intersection point while moving down in these hierarchies. This allows us to determine the intersection of the polyhedra after reaching the bottom of the hierarchies.

The algorithm described in Section~\ref{sec:algorithm-plane} directly implements this idea
to test the intersection of two convex polygons in the plane.

For technical reasons, to capture this intuition in higher dimensions we make use of the polar transformation (see Section~\ref{sec:Polar transformation}). This operation maps a polyhedron in a primal space into a dual polyhedron in a polar space.
Moreover, this transformation maps the inner hierarchy of a polyhedron into the external hierarchy of its dual counterpart. Consequently, being able to construct inner hierarchies is sufficient.
To test the intersection of two preprocessed polyhedra, our algorithm switches back and forth between a primal and a polar space while moving down in the hierarchies of these polyhedra.

\section{Algorithm in the plane}\label{sec:algorithm-plane}

\newcommand{\tp}{\ensuremath{\mathcal T_P}}
\newcommand{\tq}{\ensuremath{\mathcal T_Q}}

Let $P$ and $Q$ be two convex polygons in the plane with $n$ and $m$ vertices, respectively. 
We assume that a convex polygon is given as an array with the sequence of its vertices sorted in clockwise order along its boundary. Let $V(P)$ and $E(P)$ be the set of vertices and edges of $P$, respectively. 
Let $\partial P$ denote the boundary of $P$. Analogous definitions apply for $Q$.
As a warm-up, we describe an algorithm to determine if $P$ and $Q$ intersect whose running time is $O(\log n + \log m)$. Even though algorithms with these running time already exists in the literature, they require an involved case analysis whereas our approach avoids them and is arguably easier to implement. Moreover, it provides some intuition for the higher-dimension algorithms presented in subsequent sections.

For each edge $e\in E(Q)$, its \emph{supporting halfplane} is the halfplane containing $Q$ supported by the line extending $e$. 
Given a subset of edges $F\subseteq E(Q)$, the \emph{edge hull} of $F$ is the intersection of the supporting halfplanes of each of the edges in $F$.
Throughout the algorithm, we consider a triangle $\tp$ being the convex hull of three vertices of $P$ and a triangle (possibly unbounded) $\tq$ defined as the edge hull of three edges of $Q$; see Figure~\ref{fig:2DAlgorithm} for an illustration. Notice that $\tp\subseteq P$ while $Q\subseteq \tq$.

Intuitively, in each round the algorithm compares $\tp$ and $\tq$ for intersection and, depending on the output, prunes a fraction either of the vertices of $P$ or of the edges of $Q$. Then, the triangles $\tp$ and $\tq$ are redefined should there be a subsequent round of the algorithm. 

Let $V^*(P)$ and $E^*(Q)$ respectively be the sets of vertices and edges of $P$ and $Q$ remaining after the pruning steps performed so far by the algorithm. Initially, $V^*(P) = V(P)$ while $E^*(Q) = E(Q)$.
After each pruning step, we maintain the \emph{correctness invariant} which states that an intersection between $P$ and $Q$ can be computed with the remaining vertices  and edges after the pruning. That is, $P$ and $Q$ intersect if and only if $\ch{V^*(P)}$ intersects an edge of $E^*(Q)$, where $\ch{V^*(P)}$ denotes the convex hull of $V^*(P)$.

For a given polygonal chain, its  \emph{vertex-median} is a vertex whose removal splits this chain into two pieces that differ by at most one vertex. In the same way, the \emph{edge-median} of this chain is the edge whose removal splits the chain into two parts that differ by at most one edge.

\subsection*{The 2D algorithm}

To begin with, define $\tp$ as the convex hull of three vertices whose removal splits the boundary of $P$ into three chains, each with at most $\lceil (n-3)/3 \rceil$ vertices. In a similar way, define $\tq$ as the edge hull of three edges of $Q$ that split its boundary into three polygonal chains each with at most $\lceil (m-3)/3 \rceil$ edges; see Figure~\ref{fig:2DAlgorithm}.

A line \emph{separates} two convex polygons if they lie in opposite closed halfplanes supported by this line.
After each round of the algorithm, we maintain one of the two following invariants: The \emph{separation invariant} states that we have a line $\ell$ that separates $\tp$ from $\tq$ such that $\ell$ is tangent to $\tp$ at a vertex $v$. 
The \emph{intersection invariant} states that we have a point in the intersection between $\tp$ and $\tq$. 
Note that at least one of among separation and the
intersection invariant must hold, and they only hold at the same time when $\tp$ is tangent to $\tq$.
The algorithm performs two different tasks depending on which of the two invariants holds (if both hold, we choose a task arbitrarily).

\subsection*{Separation invariant.} If the separation invariant holds, then there is a line $\ell$ that separates $\tp$ from $\tq$ such that $\ell$ is tangent to $\tp$ at a vertex $v$. Let $\ell^-$ be the closed halfplane supported by $\ell$ that contains $\tp$ and let $\ell^+$ be its complement. 

Consider the two neighbors $n_v$ and $n'_v$ of $v$ along the boundary of $P$. Because $P$ is a convex polygon, if both $n_v$ and $n'_v$ lie in~$\ell^-$, then we are done as $\ell$ separates $P$ from $\tq\supseteq Q$.
Otherwise, by the convexity of $P$, either $n_v$ or $n'_v$ lies in $\ell^+$ but not both. Assume without loss of generality that $n_v\in \ell^+$ and notice that the removal of the vertices of $\tp$ split $\partial P$ into three polygonal chains. In this case, we know that only one of these chains, say $c_v$, intersects $\ell^+$. Moreover, we know that $v$ is an endpoint of $c_v$ and we denote its other endpoint by $u$.

Because $Q$ is contained in $\ell^+$, only the vertices in $c_v$ can define an intersection with $Q$.
Therefore, we prune $V^*(P)$ by removing every vertex of $P$ that does not lie on $c_v$ and maintain the correctness invariant.
We redefine $\tp$ as the convex hull of $v,u$ and the vertex-median of $c_v$.
With the new $\tp$, we can test in $O(1)$ time if $\tp$ and $\tq$ intersect. If they do not, then we can compute a new line that separates $\tp$ from $\tq$ and preserve the separation invariant.
Otherwise, if $\tp$ and $\tq$ intersect, then we establish the intersection invariant and proceed to the next round of the algorithm.

\subsection*{Intersection invariant.} If the intersection invariant holds, then $\tp\cap \tq\neq \emptyset$. 
In this case, let $e_1, e_2$ and $e_3$ be the three edges whose edge hull defines $\tq$.
Notice that if $\tp\subseteq P$ intersects $\ch{e_1, e_2, e_3}\subseteq Q$, then $P$ and $Q$ intersect and the algorithm finishes.
Otherwise, there are three disjoint connected components in $\tq\setminus \ch{e_1, e_2, e_3}$ and $\tp$ intersects exactly one of them; see Figure~\ref{fig:2DAlgorithm}. 
Assume without loss of generality that $\tp$ intersects the component
bounded by the lines extending $e_1$ and $e_2$ and let~$x$ be a point
on the boundary of $\tq$ in this intersection.
Let $C$ be the polygonal chain that connects $e_1$ with $e_2$ along $\partial Q$ such that $C$ passes through $e_3$. 
We claim that to test if $P$ and $Q$ intersect, we need only to consider the edges on $\partial Q\setminus C$. To prove this claim, notice that if $P$ intersects $C$ at a point $y$, then the edge $xy$ is contained in $Q$. Because $x$ and $y$ lie in two disjoint connected components of $\tq\setminus \ch{e_1, e_2, e_3}$, the edge $xy$ also intersects $\partial Q$ at another point lying on $\partial Q \setminus C$. Therefore, an intersection between $P$ and $Q$ will still be identified even if we ignore every edge on~$C$. 
That is, $P$ and $Q$ intersect if and only if $P$ and $\partial Q \setminus C$ intersect.
Thus, we can prune $E^*(Q)$ by removing every edge along $C$ while preserving the correctness invariant. 
After the pruning step, we redefine $\tq$ as the edge hull of $e_1, e_2$ and the edge-median of the remaining edges of $E(Q)$ after the pruning.

If $\tp$ intersects $\tq$ after being redefined, then the intersection invariant is preserved an we proceed to the next round of the algorithm.
Otherwise, if $\tp$ does not intersect $\tq$, then we we can compute in $O(1)$ time a line $\ell$ tangent to $\tp$ that separates $\tp$ from $\tq$. That is, the separation invariant is reestablished should there be a subsequent round of the algorithm.

\begin{figure}[t]
\centering
\includegraphics{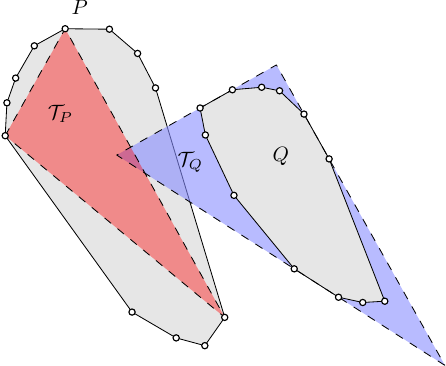}
\caption{\small Two convex polygons $P$ and $Q$ and the triangles $\tp$ and $\tq$ such that $\tq\subseteq P$ and $Q\subseteq \tq$. Moreover, $\tq\setminus Q$ consists of three connected components.}
\label{fig:2DAlgorithm}
\end{figure}

\begin{theorem}
Let $P$ and $Q$ be two convex polygons with $n$ and $m$ vertices, respectively.
The 2D-algorithm determines if $P$ and $Q$ intersect in $O(\log n + \log m)$ time.
\end{theorem}

\begin{proof}
Each time we redefine $\tp$, we take three vertices that split the remaining vertices of $V^*(P)$ into two chains of roughly equal length along $\partial P$. 
Therefore, after each round where the separation invariant holds, we prune a constant fraction of the vertices of $V^*(P)$. That is, the separation invariant step of the algorithm can be performed at most $O(\log n)$ times. 

Each time $\tq$ is redefined, we take three edges that split the remaining edges along the boundary of $Q$ into equal pieces. Thus, we prune a constant fraction of the edges of $E^*(Q)$ after each round where the intersection invariant holds. Hence, this can be done at most $O(\log m)$ times before being left with only three edges of $Q$. 
Furthermore, the correctness invariant is maintained after each of the pruning steps.

Thus, if the algorithm does not find a separating line or an intersection point, then after $O(\log n + \log m)$ steps, $\tp$ consists of the only three vertices left in $V^*(P)$ while $\tq$ consist of the only three  edges remaining from $E^*(Q)$. 
If $e_1, e_2$ and $e_3$  are the edges whose edge hull defines $\tq$, then by the correctness invariant we know that $P$ and $Q$ intersect if and only if $\tp$ intersects either $e_1, e_2$ or $e_3$.
Consequently, we can test them for intersection in $O(1)$ time and determine if $P$ and $Q$ intersect.
\end{proof}

\section{The polar transformation}\label{sec:Polar transformation}

Let $\zero$ be the \emph{origin} of $\mathbb{R}^d$, i.e., the point with $d$ coordinates equal to zero.
Throughout this paper, a \emph{hyperplane} $h$ is a $(d-1)$-dimensional affine space in $\mathbb{R}^d$ such that for some $z\in \mathbb{R}^d$, $h = \{x\in \mathbb{R}^d : \langle z, x\rangle = 1\}$, where $\langle*, *\rangle$ represents the interior product of Euclidean spaces. Therefore, in this paper a hyperplane does not contain the origin.
A \emph{halfspace} is the closure of either of the two parts into which a hyperplane divides $\mathbb{R}^d$, i.e., a halfspace contains the hyperplane defining its boundary.

Given a point $x \in \mathbb{R}^d$, we define its \emph{polar} to be the hyperplane $\polar{x}  = \{y\in \mathbb{R}^d : \langle x, y\rangle = 1\}$. Given a  hyperplane $h$ in $\mathbb{R}^d$, we define its \emph{polar} $\polar{h}$ as the point $z\in \mathbb{R}^d$ such that $h = \{y\in \mathbb{R}^d : \langle z, y\rangle = 1\}$.
Let $\planezeropolar{x} = \{y\in \mathbb{R}^d : \langle x, y\rangle \leq 1\}$ and $\planeinfpolar{x} = \{y\in \mathbb{R}^d : \langle x, y\rangle \geq 1\}$ be the two halfspaces supported by $\polar{x}$, where $\zero\in \planezeropolar{x}$ while $\zero\notin \planeinfpolar{x}$.
In the same way, $\planezero{h}$ and $\planeinf{h}$ denote the halfspaces supported by $h$ such that $\zero \in \planezero{h}$ while $\zero\notin \planeinf{h}$.

Note that the polar of a point $x\in \mathbb{R}^d$ is a hyperplane whose polar is equal to $x$, i.e., the polar operation is self-inverse (for more information on this transformation see Section 2.3 of~\cite{ziegler1995lectures}).
Given a set of points (or hyperplanes), its \emph{polar set} is the set containing the polar of each of its elements.
The following result is illustrated in Figure~\ref{fig:PolarExamples}$(a)$.

\begin{lemma}\label{lemma:polarity property}
Let $x$ and $h$ be a point and a hyperplane in $\mathbb{R}^d$, respectively.
Then, $x\in \planezero{h}$ if and only if $\polar{h}\in \planezeropolar{x}$. Also, $x\in \planeinf{h}$ if and only if $\polar{h}\in \planeinfpolar{x}$. Moreover,  $x\in h$ if and only if $\polar{h} \in \polar{x}$.
\end{lemma}
\begin{proof}
Recall that $\planezero{h} = \{y\in \mathbb{R}^d : \langle y, \polar{h}\rangle  \leq 1\}$.
Then,  $x\in \planezero{h}$ if and only if $\langle x, \polar{h}\rangle \leq 1$. 
Furthermore, $\langle x, \polar{h}\rangle \leq 1$ if and only if $\polar{h}\in \planezeropolar{x} = \{y\in \mathbb{R}^d : \langle y, x\rangle \leq 1\}$. That is, $x\in \planezero{h}$ if and only if $\polar{h}\in \planezeropolar{x}$. Analogous proofs hold for the other statements.
\end{proof}

A polyhedron is a convex region in the $d$-dimensional space being the non-empty intersection of a finite set of halfspaces.
Given a set of hyperplanes $S$ in $\mathbb{R}^d$, let $\phinf{S} = \cap_{h\in S} \planeinf{h}$ and $\phzero{S} = \cap_{h\in S} \planezero{h}$ be two polyhedra defined by $S$.
Let $P\subset \mathbb{R}^d$ be a polyhedron. 
Let $V(P)$ denote the set of vertices of $P$ and let $\shell{P}$ be the set of hyperplanes that extend the $(d-1)$-dimensional faces of $P$.
Therefore, if $P$ is bounded, then it can be seen as the convex hull of $V(P)$, denoted  by $\ch{V(P)}$. Moreover, if $P$ contains the origin, then $P$ can be also seen as $\phzero{\shell{P}}$.

To \emph{polarize} $P$, let $\dualshell{P}$ be the polar set of $V(P)$, i.e., the set of hyperplanes being the polars of the vertices of $P$. Therefore, we can think of $\phzero{\dualshell{P}}$ and $\phinf{\dualshell{P}}$ as the possible polarizations of $P$. 
For ease of notation, we let $\polarzero{P}$ and $\polarinf{P}$ denote the polyhedra $\phzero{\dualshell{P}}$ and $\phinf{\dualshell{P}}$, respectively. Note that $P$ contains the origin if and only if $\polarinf{P} = \emptyset$ and $\polarzero{P}$ is bounded.

\begin{lemma}\label{lemma:Polar of the polar}
(Clause $(v)$ of Theorem 2.11 of~\cite{ziegler1995lectures})
Let $P$ be a polyhedron in $\mathbb{R}^d$ such that $\zero \in P$. Then, $\polarzero{\polarzero{P}} = P$.
\end{lemma}

\begin{figure*}
\centering
\includegraphics{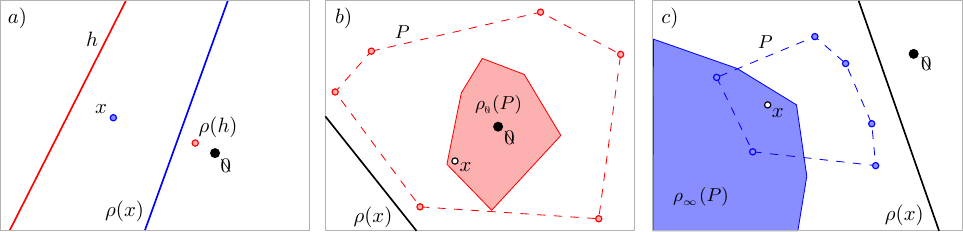}
\caption{\small $a)$ The situation described in Lemma~\ref{lemma:polarity property}. $b)$ A polygon $P$ containing the origin and its polarization $\polarzero{P}$. The first statement of Lemma~\ref{lemma:Polarizations of polyhedra} is depicted. $c)$ A polygon $P$ that does not contains the origin and its polarization $\polarinf{P}$. The second statement of Lemma~\ref{lemma:Polarizations of polyhedra} is also depicted.}
\label{fig:PolarExamples}
\end{figure*}

As a consequence of Lemma~\ref{lemma:polarity property} we obtain the following result depicted in Figures~\ref{fig:PolarExamples}$(b)$ and~\ref{fig:PolarExamples}$(c)$.

\begin{lemma}\label{lemma:Polarizations of polyhedra}
Let $P$ be a polyhedron in $\mathbb{R}^d$ and let $x\in \mathbb{R}^d$.
Then,  $x\in \polarzero{P}$ if and only if~$P\subseteq \planezeropolar{x}$.
Moreover,  $x\in \polarinf{P}$ if and only if $P\subseteq \planeinfpolar{x}$.
\end{lemma}
\begin{proof}
Let $x$ be a point in $\polarzero{P}$. 
Notice that for every hyperplane $s\in \dualshell{P}$, $x\in \planezero{s}$.
Therefore, by Lemma~\ref{lemma:polarity property} we know that the vertex $\polar{s}\in V(P)$ lies in $\planezeropolar{x}$. Consequently, every vertex of $P$ lies in $\planezeropolar{x}$, i.e., $P\subseteq \planezeropolar{x}$.

On the other direction, let $v$ be a vertex of  $P$, i.e., $\polar{v} \in \dualshell{P}$. If $v\in \planezeropolar{x}$, then by Lemma~\ref{lemma:polarity property} $x\in \planezeropolar{v}$. 
Therefore, for every $\polar{v}\in \dualshell{P}$, we know that $x\in \planezeropolar{v}$, i.e., $x\in \polarzero{P}$.

The same proof holds for the second statement by replacing all instances of $\zero$ by $\infty$.
\end{proof}

In the case that $\zero\in P$, $\polarinf{P}$ is empty and the second conclusion of the previous lemma holds trivially. Thus, even though the previous result is always true, it is non-trivial only when $\zero\notin P$.

\begin{lemma}\label{lemma:Result for Shell-simplices}
Let $P$ be a polyhedron in $\mathbb{R}^d$. 
If $x\in P$, then $\polarzero{P}\subseteq \planezeropolar{x}$ while $\polarinf{P} \subseteq \planeinfpolar{x}$.
\end{lemma}
\begin{proof}
Assume for a contradiction that there is a point $y\in \polarzero{P}$ such that $y\notin \planezeropolar{x}$.
Therefore, by Lemma~\ref{lemma:polarity property} we know that $x\notin \planezeropolar{y}$.
Moreover, because $y\in \polarzero{P}$, Lemma~\ref{lemma:Polarizations of polyhedra} implies that 
$P\subseteq \planezeropolar{y}$---a contradiction with the fact that $x\in P$ and $x\notin \planezeropolar{y}$.
An analogous proof holds to show that $\polarinf{P}\subseteq \planeinfpolar{x}$.
\end{proof}

Note that the converse of Lemma~\ref{lemma:Result for Shell-simplices} is not necessarily true.

\begin{lemma}\label{lemma:A tangent polarizes to a point inside}
Let $P$ be a polyhedron in $\mathbb{R}^d$ and let $\gamma$ be a hyperplane. 
If $\gamma$ is either tangent to $\polarzero{P}$ or to $\polarinf{P}$, then $\polar{\gamma}$ is a point lying on the boundary of $P$.
\end{lemma}
\begin{proof}
Let $\gamma$ be a hyperplane tangent to $\polarzero{P}$ at a vertex $v$. 
Because $v\in \gamma$, Lemma~\ref{lemma:polarity property} implies that $\polar{\gamma}\in \polar{v}$.
We claim that $\polar{\gamma}\in P$. Assume for a contradiction that $\polar{\gamma}\notin P$. 
Since $v\in \polarzero{P}$, we know that $P\subseteq \polarzero{v}$ by Lemma~\ref{lemma:Polarizations of polyhedra}.
Therefore, because $\polar{\gamma}\in \polar{v}$ and from the assumption that $\polar{\gamma}\notin P$, 
we can slightly perturb $\polar{v}$ to obtain a hyperplane $h$ such that $P\subseteq \planezero{h}$ while $\polar{\gamma}$ lies in the interior of~$\planeinf{h}$. 
Thus, since $\polar{\gamma}\in \planeinf{h}$ while $\polar{\gamma}\notin h$ , Lemma~\ref{lemma:polarity property} implies that $\polar{h}$ lies in the interior of $\planeinf{\gamma}$.
Moreover, because $P\subseteq \planezero{h}$ we know by Lemma~\ref{lemma:Polarizations of polyhedra} that $\polar{h}\in \polarzero{P}$.
Therefore, there is a point of $\polarzero{P}$, say $\polar{h}$, that lies in the interior of $\planeinf{\gamma}$---a contradiction with the fact that $\gamma$ is tangent to $\polarzero{P}$. Therefore, $\polar{\gamma}\in P$.
Moreover, because $\polar{\gamma}\in \polar{v}$ and from the fact that $P\subseteq \polarzero{v}$, $\polar{\gamma}$ cannot lie in the interior of $P$, i.e, $\polar{\gamma}$ lies on the boundary of $P$. An analogous proof holds for the case when $\gamma$ is tangent to $\polarinf{P}$.
\end{proof}

\begin{lemma}\label{lemma:Inverse of inclusion}
Let $P$ and $Q$ be two polyhedra. If $P\subseteq Q$, then $\polarzero{Q}\subseteq \polarzero{P}$ and $\polarinf{Q}\subseteq \polarinf{P}$.
\end{lemma}
\begin{proof}
Let $x\in \polarzero{Q}$. Then, Lemma~\ref{lemma:Polarizations of polyhedra} implies that $Q\subseteq \polarzero{x}$. 
Because we assumed that $P\subseteq Q$, $P\subseteq \polarzero{x}$. Therefore, we infer from Lemma~\ref{lemma:Polarizations of polyhedra} that $x\in \polarzero{P}$. That is,  $\polarzero{Q}\subseteq \polarzero{P}$. An analogous proof holds to show that $\polarinf{Q}\subseteq \polarinf{P}$.
\end{proof}

A hyperplane $\pi$ \emph{separates} two geometric objects in $\mathbb{R}^d$ if they are contained in opposite halfspaces supported by $\pi$, note that both objects can contain points lying on $\pi$.
We obtain the main result of this~section illustrated in Figure~\ref{fig:PolarityOfConvexPolyhedra}.

\begin{figure}[h]
\centering
\includegraphics{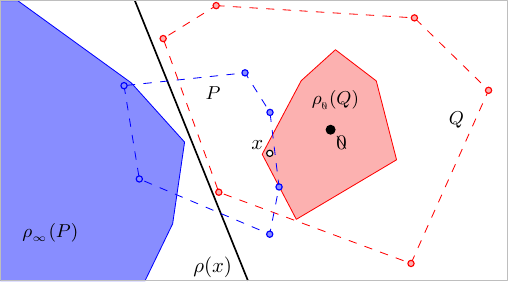}
\caption{\small The statement of Theorem~\ref{theorem: Polarity of polyhedra} where a point $x$ lies in the intersection of $P$ and $\polarzero{Q}$ if and only if $\polar{x}$ separates $Q$ from $\polarinf{P}$.}
\label{fig:PolarityOfConvexPolyhedra}
\end{figure}
\begin{theorem}\label{theorem: Polarity of polyhedra}
Let $P$ and $Q$ be two polyhedra.
The polyhedra $P$ and $\polarzero{Q}$ intersect if and only if there is a hyperplane that separates $\polarinf{P}$ from $Q$. 
Also, (1) if $x\in P\cap \polarzero{Q}$, then $\polar{x}$ separates $\polarinf{P}$ from $Q$, and (2) if $\gamma$ is a hyperplane that separates $\polarinf{P}$ from $Q$ such that $\gamma$ is tangent to $\polarinf{P}$, then $\polar{\gamma}\in P\cap \polarzero{Q}$. 
Moreover, the symmetric statements of (1) and (2) hold if we replace all instances of $P$ (\emph{resp.} $\infty$) by $Q$ (\emph{resp.} $\zero$) and vice versa.
\end{theorem}
\begin{proof}
Let $x$ be a point in $P\cap \polarzero{Q}$. Because $x \in P$, by Lemma~\ref{lemma:Result for Shell-simplices} we know that 
$\polarinf{P}\subseteq \planeinfpolar{x}$.
Moreover, since $x \in \polarzero{Q}$, by Lemma~\ref{lemma:Polarizations of polyhedra}, 
$Q\subseteq \planezeropolar{x}$. Therefore, $\polar{x}$ is a hyperplane that separates $\polarinf{P}$ from $Q$. 

In the other direction, let $\gamma'$ be a hyperplane that separates $\polarinf{P}$ from $Q$. 
Then, there is a hyperplane $\gamma$ parallel to $\gamma'$ that separates $\polarinf{P}$ from $Q$ such that $\gamma$ is tangent to $\polarinf{P}$. 
Therefore, $\polar{\gamma}$ is a point on the boundary of $P$ by Lemma~\ref{lemma:A tangent polarizes to a point inside}.
Because $\polar{\gamma}\in P$, Lemma~\ref{lemma:Result for Shell-simplices} implies that $\polarzero{P}\subseteq \planezero{\gamma}$ while $\polarinf{P}\subseteq \planeinf{\gamma}$. Because $\gamma$ separates $\polarinf{P}$ from $Q$ and from the fact that $\polarinf{P}\subseteq \planeinf{\gamma}$, we conclude that $Q\subseteq \planezero{\gamma}$.
Consequently, by Lemma~\ref{lemma:Polarizations of polyhedra} $\polar{\gamma}\in \polarzero{Q}$. 
That is, $\polar{\gamma}$ is a point in the intersection of $P$ and $\polarzero{Q}$.
The symmetric statements have analogous proofs.
\end{proof}

Notice that if $\zero \in P$, then $P$ and $\polarzero{Q}$ trivially intersect. Moreover, $\polarinf{P} = \emptyset$ implying that every hyperplane trivially separates $\polarinf{P}$ from $Q$. Therefore, while being always true, this result is non-trivial only when~$\zero \notin P$.

\section{Polyhedra in 3D space}\label{section:Bounded degree polyhedron}

In this section, we focus on polyhedra in $\mathbb{R}^3$. Therefore, we can consider the 1-skeleton of a polyhedron being the planar graph connecting its vertices through the edges of the polyhedron.

Given a polyhedron $P$, a sequence $P_1, P_2, \ldots, P_k$ is a DK-hierarchy of $P$ if the following properties hold~\cite{dobkin1985linear}.
\begin{enumerate}
\item[$A1.$] $P_1 = P$ and $P_k$ a tetrahedron.
\item[$A2.$] $P_{i+1} \subseteq P_i$, for $1\leq i\leq k$.
\item[$A3.$] $V(P_{i+1}) \subseteq V(P_i)$, for $1\leq i\leq k$.
\item[$A4.$] The vertices of $V(P_i)\setminus V(P_{i+1})$ form an independent set in $P_i$, for $1\leq i< k$.
\item[$A5.$] The \emph{height} of the hierarchy $k = O(\log n)$, $\sum_{i=1}^k V(P_i) = O(n)$.
\end{enumerate}

Given a polyhedron $P$ on $n$ vertices, a set $I\subseteq V(P)$ is a \emph{$P$-independent set} if (1) $|I| \geq n/10$, (2) $I$ forms an independent set in the 1-skeleton of $P$ and (3) the degree of every vertex in $I$ is~$O(1)$.

Dobkin and Kirkpatrick~\cite{dobkin1985linear} showed how to construct a DK-hierarchy. This construction was later improved by Biedl and Wilkinson~\cite{biedl2005bounded}.
Formally, they start by defining $P_1 = P$. Then, given a polyhedron $P_i$, they show how to compute a $P_i$-independent set~$I$ and define $P_{i+1}$ as the convex hull of the set $V(P_i)\setminus I$. 

Using this data structure, they provide an algorithm that computes the distance between two preprocessed polyhedra in $O(\log^2 n)$ time~\cite{dobkin1990determining}. 
As we show below however, a straightforward implementation
of their algorithm could be be much slower than this claimed bound. 

In our algorithm, as well as in the algorithm presented by Dobkin and Kirkpatrick~\cite{dobkin1990determining}, we are given a plane tangent to $P_i$ at a vertex $v$ and want to find a vertex of $P_{i-1}$ lying on the other side of this plane (if it exists). 
Although they showed that at most one vertex of $P_{i-1}$ can lie on the other side of this plane and that it has to be adjacent to $v$, they do not explain how to find such a vertex. 
An exhaustive walk through the neighbors of
$v$ in $P_{i-1}$ would only be fast enough for their algorithm if $v$ is always of
constant degree. Unfortunately this is not always the case as shown
in the following example.

Start with a tetrahedron $P_k$ and select a vertex $q$ of $P_k$. To construct the polyhedron $P_{i-1}$ from $P_i$, we refine it by adding a vertex slightly above each face adjacent to $q$. In this way, the degree of the new vertices is exactly three. 
After $k$ steps, we reach a polyhedron $P_1 = P$. 
In this way, the sequence $P = P_1, P_2, \ldots, P_k$ defines a DK-hierarchy of $P$. 
Moreover, when going from $P_i$ to $P_{i-1}$, a new neighbor of $q$ is added for each of its adjacent faces in $P_i$. Thus, the degree of $q$ doubles when going from $P_i$ to $P_{i-1}$ and hence, the degree of $q$ in $P_1$ is linear.
Note that this situation can occur at a deeper level of the hierarchy, even if every vertex of $P$ has degree three.

A solution to this problem is described by O'Rourke~\cite[Chapter 7]{o1998computational}. In the next section, we provide an alternative solution to this problem by bounding the degree of each vertex in every polyhedron of the DK-hierarchy.

\begin{figure*}
\centering
\includegraphics{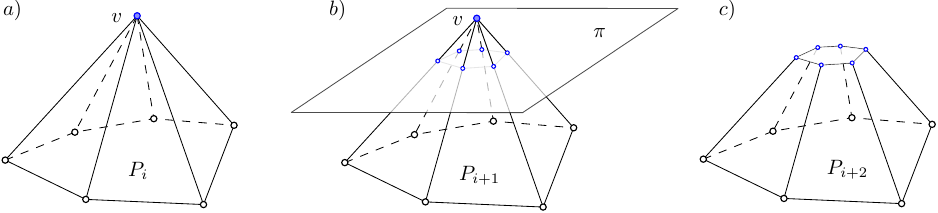}
\caption{\small A polyhedron $P$ and a vertex $v$ of large degree.  A plane $\pi$ that separates $v$ from $V(P)\setminus \{v\}$ is used to split the edges adjacent to $v$. New vertices are added to split these edges.  Finally, the removal of $v$ from the polyhedron leaves every one of its neighbors with degree three while adding a new face.}
\label{fig:Chopping}
\end{figure*}

\subsection*{Bounded hierarchies}

Let $c$ be a fixed constant.
We say that a polyhedron is \emph{\bounded} if at most $c$ faces of this polyhedron can meet at a vertex, i.e., the degree of each vertex in its 1-skeleton is bounded by~$c$.

Given a polyhedron $P$ with $n$ vertices, we describe a method to modify the structure of Dobkin and Kirkpatrick to construct a DK-hierarchy where every polyhedron other than $P$ is $c$-bounded.
As a starting point, we can assume that the faces of $P$ are in general position (i.e., no four planes of $\shell{P}$ go through a single point) by using Simulation of Simplicity~\cite{edelsbrunner1990simulation}. This implies that every vertex of $P$ has degree three.
To avoid having vertices of large degree in the hierarchy, we introduce the following operation.
Given a vertex $v\in V(P)$ of degree $k>3$, consider a plane $\pi$ that separates $v$ from every other vertex of $P$. 
Let $e_1, e_2, \ldots, e_k$ be the edges of $P$ incident to $v$. For each $1\leq i\leq k$, let $v_i$ be the intersections of $e_i$ with $\pi$. \emph{Split} the edge $e_i$ at $v_i$ to obtain a new polyhedron with $k$ more vertices and $k$ new edges; for an illustration see Figure~\ref{fig:Chopping} $(a)$ and $(b)$.

To construct a \bounded DK-hierarchy (or simply \hier), we start by letting $P_1 = P$. 
Given a polyhedron $P_i$ in this \hier, let $I$ be a $P_i$-independent set.
Compute the convex hull of $V(P_i) \setminus I$, two cases arise:
\textbf{Case 1.} If $\ch{V(P_i) \setminus I}$ has no vertex of degree larger than $c$, then let $P_{i+1} = \ch{V(P_i) \setminus I}$. 
\textbf{Case 2.} Otherwise, let $W$ be the set of vertices of $P_i$ with degree larger than $3$. For each vertex of $W$, split its adjacent edges as described above and let $P_{i+1}$ be the obtained polyhedron.
Notice that $P_{i+1}$ is a polyhedron with the same number of faces as $P_i$. 
Moreover, because each edge of $P_i$ may be split for each of its endpoints, 
$P_{i+1}$ has at most three times the number the edges of $P_i$. 
Therefore $|V(P_{i+1})| \leq (2/3)|E(P_{i+1}) \leq 2|E(P_i)| \leq 6|V(P_i)|$ by Euler's formula.

Because each vertex of $W$ is adjacent only to new vertices added during the split of its adjacent edges, the vertices in $W$ form an independent set in the 1-skeleton of $P_{i+1}$.
In this case, we let $P_{i+2}$ be the convex hull of $V(P_{i+1}) \setminus W$. Therefore, $(1)$ every vertex of $P_{i+2}$ has degree three, and $(2)$ the vertices in $V(P_{i+1})\setminus V(P_{i+2})$ form an independent set in $P_{i+1}$; see Figure~\ref{fig:Chopping}$(c)$.
Note that $P_{i+1}$ and $P_{i+2}$ have new vertices added during the splits. However, we know that $|V(P_{i+2})| \leq |V(P_{i+1})| \leq 6|V(P_i)|$. 
Furthermore, we also know that $P_{i+2}\subseteq P_{i+1}\subseteq P_i$.

We claim that by choosing $c$ carefully, we can guarantee that the depth of the \hier is $O(\log n)$.
To prove this claim, notice that after a pruning step, the
degree of a vertex can increase at most by the total degree of its
neighbors that have been eliminated.
Let $v$ be a vertex with the largest degree in $P_i$. Note that its neighbors can also have at most degree $\delta(v)$, where $\delta(v)$ denotes the number of neighbors of $v$ in $P_i$. Therefore, after removing a $P_i$-independent set, the degree of $v$ can be at most $\delta(v)^2$ in $P_{i+1}$. 
That is, the maximum degree of $P_i$ can be at most squared when going from $P_i$ to $P_{i+1}$. 

Therefore, if we assume Case~2 has just been applied and that every vertex vertex of $P_i$ has degree three, then after $r$ rounds of Case~1, the
maximum degree of any vertex is at most $3^{2^r}$.
Therefore, the degree of any of its vertices can go above $c$ only after $\log_2(\log_3 c)$ rounds, i.e., we go through Case 1 at least $\log_2(\log_3 c)$ times before running into Case 2.

Since we removed at least $1/10$-th of the vertices after each iteration of Case 1~\cite{biedl2005bounded}, after $\log_2(\log_3 c)$ rounds the size of the current polyhedron is at most $(9/10)^{\log_2(\log_3 c)}|P_i|$. 
At this point, we run into Case 2 and add extra vertices to the polyhedron. However, by choosing $c$ sufficiently large, we guarantee that the number of remaining vertices is at most $6\cdot (9/10)^{\log_2(\log_3 c)}|P_i| < \alpha |P_i|$ for some constant $0<\alpha<1$. That is, after $\log_2(\log_3 c)$ rounds the size of the polyhedron decreases by constant factor implying a logarithmic depth. We obtain the following result.

\begin{lemma}
Given a polyhedron $P$, the previous algorithm constructs a \hier $P_1, P_2, \ldots, P_k$ with following properties.
\begin{enumerate}
\item[$B1.$] $P_1 = P$ and $P_k$ is a tetrahedron.
\item[$B2.$] $P_{i+1} \subseteq P_i$, for $1\leq i\leq k$.
\item[$B3.$] The polyhedron $P_i$ is \bounded, for $1\leq i\leq k$.
\item[$B4.$] The vertices of $V(P_i)\setminus V(P_{i+1})$ form an independent set in $P_i$, for $1\leq i< k$.\label{property:Independence}
\item[$B5.$] The \emph{height} of the hierarchy $k = O(\log n)$, $\sum_{i=1}^k V(P_i) = O(n)$.
\end{enumerate}
\end{lemma}


The following property of a DK-hierarchy of $P$ was proved in~\cite{dobkin1990determining} and is easily extended to \hiers because its proof does not use property $A3$. Note that all
properties of DK and BDK hierarchies are identical except for $B3\neq A3$.

\begin{lemma}\label{lemma:halfspace property}
Let $P_1, \ldots, P_k$ be a \hier of a polyhedron $P$ and let $H$ be a plane defining two halfspaces $H^+$ and $H^-$.
For any $1\leq i\leq k$ such that $P_{i+1}$ is contained in $H^+$, either $P_i\subseteq H^+$ or there exists a unique vertex $v\in V(P_i)$ such that $v\in H^-\setminus H$.
\end{lemma}

\section{Detecting intersections in 3D}\label{section:3D algorithm}

In this section, we show how to independently preprocess polyhedra in 3D-space so that their intersection can be tested in logarithmic time.
\subsection*{Preprocessing}
Let $P$ be a polyhedron in $\mathbb{R}^3$. Assume without loss of generality that the origin lies in the interior of $P$. Otherwise, modify the coordinate system.
To preprocess $P$, we first  compute the polyhedron $\polarzero{P}$ being the polarization of $P$. Then, we independently compute two \hiers as described in Section~\ref{section:Bounded degree polyhedron}, one for $P$ and one for $\polarzero{P}$. 
Recall that in the construction of \hiers, we assume that the faces of the polyhedra being processed are in general position using Simulation of Simplicity~\cite{edelsbrunner1990simulation}.
Assuming that both $P$ and $\polarzero{P}$ have vertices in general position at the same time is not possible. However, this is not a problem as only one of the two \hiers will ever be used in a single query. Therefore, we can independently use Simulation of Simplicity~\cite{edelsbrunner1990simulation} on each of them.

\subsection*{Preliminaries of the algorithm}
Let $P$ and $R$ be two independently preprocessed polyhedra with combinatorial complexities $n$ and $m$, respectively. Throughout this algorithm, we fix  the coordinate system used in the preprocessing of $R$, i.e., $\zero\in R$. For ease of notation, let $Q = \polarzero{R}$. Because $\zero\in R$, Lemma~\ref{lemma:Polar of the polar} implies that $R = \polarzero{Q}$.

The algorithm described in this section tests if $P$ and $R = \polarzero{Q}$ intersect. Therefore, we can assume that $P$ and $\polarzero{Q}$ lie in a \emph{primal space} while $\polarinf{P}$ and $Q$ lie in a \emph{polar space}. 
That is, we look at the primal and polar spaces independently and switch between them whenever necessary.
To test the intersection of $P$ and $\polarzero{Q}$ in the primal space, we use the \hiers of $P$ and $Q$ stored in the preprocessing step.
In an intersection query, we are given arbitrary translations and rotations for $P$ and $\polarzero{Q}$ and we want to decide if they intersect. Note that this is equivalent to answering the query when only a translation and rotation of $P$ is given and $\polarzero{Q}$ remains unchanged. This is important as we fixed the position of the origin inside $R = \polarzero{Q}$.
The idea of the algorithm is to proceed by rounds and in each of them, move down in one of the two hierarchies while maintaining some invariants. In the end, when reaching the bottom of the hierarchy, we determine if $P$ and $\polarzero{Q}$ are separated or not.

Let $k$ and $l$ be the depths of the hierarchies of $P$ and $Q$, respectively.
We use indices $1\leq i\leq k$ and $1\leq j\leq l$ to indicate our position in the hierarchies of $P$ and $Q$.
The idea is to decrement at least one of them in each round of the algorithm.

To maintain constant time operations, instead of considering a full polyhedron $P_i$ in the \hier of $P$, we consider constant complexity polyhedra $P_i^*\subseteq P_i$ and $Q_j^*\subseteq Q_j$. 
Intuitively, both $P_i^*$ and $Q_j^*$ are constant size polyhedra that respectively represent the portions of $P_i$ and $Q_j$ that need to be considered to test for an intersection.

We also maintain a special point $\pointP$ in the primal space which is a vertex of both $P_i^*$ and $P_i$, and a plane $\sepP$ whose properties will be determined later. In the polar space, we keep a point $\pointQ$ being a vertex of both $Q_j^*$ and $Q_j$ and a plane $\sepQ$.

For ease of notation, given a polyhedron $T$ and a vertex $v\in V(T)$, let $T\setminus v$ denote the convex hull of $V(T)\setminus \{v\}$.
The \emph{star invariant} consists of two parts, one in the primal and another in the polar space.
In the primal space, this invariant states that if $i< k$, then (1) the plane $\sepP$ separates $P_i\setminus \pointP$ from $\polarzero{Q_j}$ and (2) $\polar{\sepP}\in Q_j$.
In the polar space, the star invariant states if $j < l$, then (1) the plane $\sepQ$ separates $Q_j\setminus \pointQ$ from $\polarinf{P_i}$ and (2) $\polar{\sepQ}\in P_i$.
Whenever the star invariant is established, we store references to $\sepP$ and $\sepQ$, and to the vertices $\pointP$ and $\pointQ$.

Other invariants are also considered throughout the algorithm.
The \emph{separation invariant} states that we have a plane $\pi$ that separates $P_i$ from $\polarzero{Q_j^*}$ such that $\pi$ is tangent to $P_i$ at one of its vertices.
The \emph{inverse separation invariant} states that there is a plane $\mu$ that separates $\polarinf{P_i^*}$ from $Q_j$ such that $\mu$ is tangent to~$Q_j$ at one of its vertices.

Before stepping into the algorithm, we need a couple of definitions. Given a polyhedron $T$ and a vertex $v\in V(T)$, let $\neighbors{v}{T}$ be a polyhedron defined as the convex hull of $v$ and its neighbors in~$T$.
Let $\cone{v}{T}$ be the convex hull of the set of rays apexed at $v$ shooting from $v$ to each of its neighbors in $T$. That is, $\cone{v}{T}$ is a convex cone with apex $v$ that contains $T$ and has complexity $O(\delta(v))$, where $\delta(v)$ denotes the number of neighbors of $v$ in $T$. We say that $\cone{v}{T}$ \emph{separates} $T$ from another polyhedron if the latter does not intersect the interior of $\cone{v}{T}$.

\subsection*{The algorithm}

To begin the algorithm, let $i = k$ and  $j = l$, i.e., we start with $P_i^* = P_i$ and $Q_j^* = Q_j$ being both tetrahedra.
Notice that for the base case, $i = k$ and $j = l$, we can determine in $O(1)$ time if $P_i$ and $\polarzero{Q_j} = \polarzero{Q_j^*}$ intersect. 
If they do not, then we can compute a plane separating them and establish the separation invariant. 
Otherwise, if $P_i$ and $\polarzero{Q_j}$ intersect, then by Theorem~\ref{theorem: Polarity of polyhedra} we know that $\polarinf{P_i} = \polarinf{P_i^*}$ does not intersect $Q_j$.
Thus, in constant time we can compute a plane tangent to $Q_j$ in the polar space that separates $\polarinf{P_i} = \polarinf{P_i^*}$ from $Q_j$. That is, we can establish the inverse separation invariant. Thus, at the beginning of the algorithm the star invariant holds trivially, and either the separation invariant or the inverse separation invariant holds (maybe both if $P_i$ and $\polarzero{Q_j}$ are tangent).

After each round of the algorithm, we advance in at least one of the hierarchies of $P$ and $Q$ while maintaining the star invariant. Moreover, we maintain at least one among the separation and the inverse separation invariants. Depending on which invariant is maintained, we step into the primal or the polar space as follows (if both invariants hold, we choose arbitrarily).

\subsection*{A walk in the primal space.} 
We step into this case if the separation invariant holds.
That is, $P_i$ is separated from $\polarzero{Q_j^*}$ by a plane $\pi$ tangent to $P_i$ at a vertex $v$. 

We know by Lemma~\ref{lemma:halfspace property} that there is at most one vertex $p$ in $P_{i-1}$ that lies in $\planezero{\pi}\setminus \pi$.
Moreover, this vertex must be a neighbor of $v$ in $P_{i-1}$. Because $P_{i-1}$ is \bounded, we scan the $O(1)$ neighbors of $v$ and test if any of them lies in $\planezero{\pi}\setminus \pi$. Two cases arise:

\vspace{.1in}
\textbf{Case 1.} If $P_{i-1}$ is contained in $\planeinf{\pi}$, then $\pi$ still separates $P_{i-1}$ from $\polarzero{Q_j^*}$ while being tangent to the same vertex $v$ of $P_{i-1}$. Therefore, we have moved down one level in the hierarchy of $P$ while maintaining the separation invariant. 

To maintain the star invariant, let $P_{i-1}^* = \neighbors{v}{P_{i-1}}$ and let $\pointP = v\in V(P_{i-1}^*)\cap V(P_{i-1})$. 
Because $P_{i-1}$ is $c$-bounded, we know that $P_{i-1}^*$ has constant size. 
Since $\polarzero{Q_j^*}$ has constant size, we can compute the plane $\sepP$ parallel to $\pi$ and tangent to $\polarzero{Q_j^*}$ in $O(1)$ time, i.e.,
$\sepP$ also separates $P_{i-1}$ from $\polarzero{Q_j^*}$.
Because $\polarzero{Q_j^*}\supseteq \polarzero{Q_j}$ by Lemma~\ref{lemma:Inverse of inclusion} and from the fact that 
$P_{i-1}\setminus \pointP\subset P_{i-1}$, we conclude that (1) $\sepP$ separates $P_{i-1}\setminus \pointP$ from $\polarzero{Q_j}$. Moreover, because $\polar{\sepP}\in Q_j^*$ by Lemma~\ref{lemma:A tangent polarizes to a point inside} and from the fact that $Q_j^*\subseteq Q_j$, we conclude that (2) $\polar{\sepP}\in Q_j$. Thus, the star invariant is maintained in the primal space.

In the polar space, if $j< l$, then since $\polarinf{P_{i-1}}\subseteq \polarinf{P_i}$ by Lemma~\ref{lemma:Inverse of inclusion}, (1) the plane $\sepQ$ that separates $Q_j\setminus \pointQ$ from $\polarinf{P_i}$ also separates $Q_j\setminus \pointQ$ from $\polarinf{P_{i-1}}$. 
Moreover, because $P_i\subseteq P_{i-1}$ and from the fact that $\polar{\sepQ}\in P_i$, we conclude that (2) $\polar{\sepQ}\in P_{i-1}$.
Thus, the star invariant is also maintained in the polar space and we proceed with a new round of the algorithm in the primal~space.

\vspace{.1in}
\textbf{Case 2.} If $P_{i-1}$ crosses $\pi$, then by Lemma~\ref{lemma:halfspace property} there is a unique vertex $p$ of $P_{i-1}$ that lies in~$\planezero{\pi}\setminus \pi$. 
To maintain the star invariant, let $P_{i-1}^* = \neighbors{p}{P_{i-1}}$ and let $\pointP = p$. Then, proceed as in to the first case. 
In this way, we maintain the star invariant in both the primal and the polar space.

Recall that $\cone{\pointP}{P_{i-1}}$ is the cone being the convex hull of the set of rays shooting from $\pointP$ to each of its neighbors in $P_{i-1}$.
Since $P_{i-1}$ is $c$-bounded, $\pointP$ has at most $c$ neighbors in $P_{i-1}$. 
Thus, both $\cone{\pointP}{P_{i-1}}$ and $\polarzero{Q_j^*}$ have constant complexity and we can test if they intersect in constant time.
Two cases arise: 

\vspace{.1in}
\textbf{Case 2.1.} If $\cone{\pointP}{P_{i-1}}$ and $\polarzero{Q_j^*}$ do not intersect, then as $P_{i-1}\subseteq \cone{\pointP}{P_{i-1}}$, we can compute in constant time a plane $\pi'$ tangent to $\cone{\pointP}{P_{i-1}}$ at $\pointP$ that separates $P_{i-1}\subseteq \cone{\pointP}{P_{i-1}}$ from $\polarzero{Q_j^*}$.
That is, we reestablish the separation invariant and proceed with a new round in the primal space.\vspace{.1in}

\textbf{Case 2.2.} Otherwise, if $\cone{\pointP}{P_{i-1}}$ and $\polarzero{Q_j^*}$ intersect, then because $P_{i-1}\setminus \pointP \subseteq \planeinf{\pi}$ and $\polarzero{Q_j^*}\subseteq \planezero{\pi}$, we know that this intersection happens at a point of $P_{i-1}^*$, i.e., $P_{i-1}^*$ intersects $\polarzero{Q_j^*}$.
Therefore, by Theorem~\ref{theorem: Polarity of polyhedra} there is a plane $\mu'$ that separates $\polarinf{P_{i-1}^*}$ from $Q_j^*$ in the polar space. In this case, we would like to establish the inverse separation invariant which states that $\polarinf{P_{i-1}^*}$ is separated from $Q_j$.
Note that if $j = l$, then $Q_j = Q_j^*$ and the inverse separation invariant is established. 
Therefore, assume that $j<l$ and recall that $\pointQ\in V(Q_j^*)\cap V(Q_j)$.

By the star invariant and from the assumption that $j<l$, the plane $\sepQ$ separates $Q_j\setminus \pointQ$ from $\polarinf{P_{i-1}}$, i.e., $Q_j\setminus \pointQ\subseteq \planezero{\sepQ}$.
In this case, we enlarge $P_{i-1}^*$ by adding the vertex $\polar{\sepQ}$ to it, i.e., we let $P_{i-1}^* = \ch{\neighbors{p}{P_{i-1}}\cup\{ \polar{\sepQ}\}}$. 
Note that this enlargement preserves the star invariant as $\pointP$ is still a vertex of the refined $P_{i-1}^*$. Moreover, because $\polar{\sepQ}\in P_{i-1}$ by the star invariant, we know that $P_{i-1}^*\subseteq P_{i-1}$.

Because $\polar{\sepQ}\in P_{i-1}^*$, Lemma~\ref{lemma:Result for Shell-simplices} implies that $\polarinf{P_{i-1}^*}\subseteq \planeinf{\sepQ}$. 
Since $Q_j\setminus \pointQ\subseteq \planezero{\sepQ}$, $\sepQ$ separates $\polarinf{P_{i-1}^*}$ from $Q_j\setminus \pointQ$. 
Because $\mu'$ separates $\polarinf{P_{i-1}^*}$ from~$Q_j^*\supseteq \neighbors{\pointQ}{Q_j}$, we conclude that there is a plane that separates $\polarinf{P_{i-1}^*}$ from $Q_j$ and it only remains to compute~it in $O(1)$ time. 

In fact, because $Q_j\setminus \pointQ\subseteq \planezero{\sepQ}$, all neighbors of $\pointQ$ in $Q_j$ lie in $\planezero{\sepQ}$ and hence, the cone $\cone{\pointQ}{Q_j}$ does not intersect $\polarinf{P_{i-1}^*}$.
Since $\cone{\pointQ}{Q_j}$ and $\polarinf{P_{i-1}^*}$ have constant complexity, we can compute a plane $\mu$ tangent to $\cone{\pointQ}{Q_j}$ at $\pointQ$ such that $\mu$ separates $\cone{\pointQ}{Q_j}$ from $\polarinf{P_{i-1}^*}$. 
Because $Q_j\subseteq \cone{\pointQ}{Q_j}$, $\mu$ separates $Q_j$ from $\polarinf{P_{i-1}^*}$ while being tangent to $Q_j$ at $\pointQ$.
That is, we establish the inverse separation invariant.
In this case, we step into the polar space and try to move down in the hierarchy of $Q$ in the next round of the algorithm.

\subsection*{A walk in the polar space.}
We step into this case if the inverse separation invariant holds. That is, we have a plane tangent to $Q_j$ at one of its vertices that separates $\polarinf{P_i^*}$ from $Q_j$.
In this case, we perform an analogous procedure to that described for the case when the separation invariant holds. However, all instances of $P_i$ (\emph{resp.} $P$) are replaced by $Q_j$ (\emph{resp.} $Q$) and vice versa, and all instances of $\polarinf{*}$ are replaced by $\polarzero{*}$ and vice versa. Moreover, all instances of the separation and the inverse separation invariant are also swapped. At the end of this procedure, we decrease the value of $j$ and establish either the separation or the inverse separation invariant. Moreover, the star invariant is also preserved should there be a subsequent round of the algorithm.

\subsection*{Analysis of the algorithm}

After going back and forth between the primal and the polar space, we reach the bottom of the hierarchy of either $P$ or $Q$. 
Thus, we might reach a situation in which we analyze $P_1$ and~$\polarzero{Q_j^*}$~in the primal space for some $1\leq j\leq l$. In this case, if the separation invariant holds, then we have computed a plane $\pi$ that separates $P_1$ from $\polarzero{Q_j^*}\supseteq \polarzero{Q}$. 
Because $P = P_1$, we conclude that $\pi$ separates $P$ from $R = \polarzero{Q}$.

We may also reach a situation in which we test $Q_1$ and $\polarinf{P_i^*}$ in the polar space for some $1\leq i\leq k$. 
In this case, if the inverse separation invariant holds, then we have a plane $\mu$ that separates~$Q_1$ from $\polarinf{P_i^*}$. 
Since $\polarinf{P_i^*}$ has constant complexity, we can assume that $\mu$ is tangent to $\polarinf{P_i^*}$ as we can compute a plane parallel to $\mu$ with this property.
Because $Q = Q_1$, we conclude that $\mu$ is a plane that separates $Q$ from $\polarinf{P_i^*}$ such that $\mu$ is tangent to $\polarinf{P_i^*}$.
Therefore, Theorem~\ref{theorem: Polarity of polyhedra} implies that $\polar{\mu}$ is a point in the intersection of $P_i^*\subseteq P$ and $\polarzero{Q}$, i.e., $P$ and $R = \polarzero{Q}$ intersect.

In any other situation the algorithm can continue until one of the two previously mentioned cases arises and the algorithm finishes. 
Because we advance in each round in either the \hier of $P$ or the \hier of $Q$, after $O(\log n + \log m)$ rounds the algorithm finishes. Because each round is performed in $O(1)$ time, we obtain the following result.

\begin{theorem}
Let $P$ and $R$ be two independently preprocessed polyhedra in $\mathbb{R}^3$ with combinatorial complexities $n$ and $m$, respectively. 
For any given translations and rotations of $P$ and~$R$, we can determine if $P$ and~$R$ intersect in $O(\log n + \log m)$ time.
\end{theorem}

\section{Detecting intersections in higher dimensions}\label{sec:high-d}

In this section, we extend our algorithm to any constant dimension $d$ at the expense of increasing the space to $O(n^{\lfloor d/2\rfloor + \delta})$ for any $\delta>0$. 
To do that, we replace the \hier and introduce a new hierarchy produced by recursively taking $\varepsilon$-nets of the faces of the polyhedron. Our objective is to obtain a new hierarchy with logarithmic depth with properties
similar to those described in Lemma~\ref{lemma:halfspace property}. For the latter, we use the following definition.

Given a polyhedron $P$, the intersection of $(d+1)$ halfspaces is a \emph{shell-simplex} of $P$ if it contains $P$ and each of these $(d+1)$ halfspaces is supported by a $(d-1)$-dimensional face of $P$.

\begin{lemma}\label{lemma:Simplex lemma}
Let $P$ be a polyhedron in $\mathbb{R}^d$ with $k$ vertices. We can compute a set $\Sigma(P)$ of at most $O(k^{\lfloor d/2\rfloor})$ shell-simplices of~$P$ such that given a hyperplane $\pi$ tangent to $P$, there is a shell-simplex $\sigma\in \Sigma(P)$ such that $\pi$ is also tangent to $\sigma$.
\end{lemma}
\begin{proof}
Without loss of generality assume that $\zero\in P$. 
Note that $\polarzero{P}$ has exactly $k$ $(d-1)$-dimensional faces.
Using Lemma 3.8 of~\cite{clarkson1988randomized} we infer that there exists a triangulation $T$ of $\polarzero{P}$ such that the combinatorial complexity of $T$ is $O(k^{\lfloor d/2\rfloor})$. That is, $T$ decomposes $\polarzero{P}$ into interior disjoint $d$-dimensional simplices.

Let $s$ be a simplex of $T$.
For each $v\in V(s)$, notice that since $v\in \polarzero{P}$, $P\subseteq \polarzero{v}$ by Lemma~\ref{lemma:Result for Shell-simplices}.
Therefore, $P\subseteq \cap_{v\in V(s)} \polarzero{v} = \polarzero{s}$, i.e., $\sigma_s = \polarzero{s}$ is a shell-simplex of $P$ obtained from polarizing $s$. 
Finally, let $\Sigma(P) = \{\sigma_s : s\in T\}$ and notice that~$|\Sigma(P)| = O(k^{\lfloor d/2\rfloor})$.

Because $\zero\in P$, Lemma~\ref{lemma:Polar of the polar} implies that $P = \polarzero{\polarzero{P}}$.
Let $\pi$ be a hyperplane tangent to $P = \polarzero{\polarzero{P}}$ and note that its polar is a point $\polar{\pi}$ lying on the boundary of $\polarzero{P}$ by Lemma~\ref{lemma:A tangent polarizes to a point inside}.
Hence, $\polar{\pi}$ lies on the boundary of a simplex~$s$ of~$T$. Thus, by Lemma~\ref{lemma:Result for Shell-simplices} 
we know that $\sigma_s\subseteq \planezero{\pi}$.
Because $\polar{\pi}$ lies on the boundary of $s$, $\pi$ is tangent to $\sigma_s$ yielding our result.
\end{proof}

\subsection*{Hierarchical trees}

Let $P$ be a polyhedron  with combinatorial complexity $n$. 
We can assume that the vertices of $P$ are in general position (i.e., no $d+1$ vertices lie on the same hyperplane) using Simulation of Simplicity~\cite{edelsbrunner1990simulation}.

Let $F(P)$ be the set of all faces of $P$.
Consider the family $G$ such that a set $g\in G$ is the complement of the intersection of $d+1$ halfspaces.
Let $F_{g} = \{f\in F(P) : f\cap g\neq \emptyset\}$ be the set of faces of $P$ induced by $g$.
Let $G_{F(P)} = \{F_g : g \in G\}$ be the family of subsets of $F(P)$ induced by $G$.

To compute the hierarchy of $P$, let $0< \varepsilon < 1$ and consider the range space defined by $F(P)$ and $G_{F(P)}$.
Since the VC-dimension of this range space is finite, we can compute an $\varepsilon$-net $N$ of $(F(P), G_{F(P)})$ of size $O(\frac{1}{\varepsilon} \log \frac{1}{\varepsilon}) = O(1)$~\cite{ConstructionEpsilonNets}. 
Because the vertices of $P$ are in general position, each face of $P$ has at most $d+1$ vertices. Therefore, $\ch{N}$ has $O(|N|)$ vertices, i.e., $\ch{N}$ has constant complexity.
By Lemma~\ref{lemma:Simplex lemma} and since $|N| = O(1)$, we can compute the set $\Sigma(\ch{N})$ having $O(|N|^{\lfloor d/2\rfloor})$ shell-simplices of $\ch{N}$ in constant time.

Given a shell-simplex $\sigma\in \Sigma(\ch{N})$, let $\bar{\sigma}\in  G$ be the complement of $\sigma$. 
Because $\ch{N}\subseteq \sigma$, $\bar{\sigma}$ intersects no face of $N$. 
Recall that $F_{\bar{\sigma}} = \{f\in F(P) : f\cap \bar{\sigma}\neq \emptyset\}$.
Therefore, since $N$ is an $\varepsilon$-net of $(F(P), G_{F(P)})$, we conclude that  $F_{\bar{\sigma}}$ contains at most $\varepsilon |F(P)|$ faces of $P$. 

We construct the \emph{hierarchical tree} of a polyhedron $P$ recursively.  In each recursive step, we consider a subset $F$ of the faces of $P$ as input.
As a starting point, let $F = F(P)$.
The recursive construction considers two cases: (1) If $F$ consists of a constant number of faces, then its tree consists of a unique node storing a reference to $\ch{F}$. (2) Otherwise, compute the $\varepsilon$-net $N$ of $F$ as described above and store $\ch{N}$ together with $\Sigma(\ch{N})$ at the root node.
Then, for each shell-simplex $\sigma\in \Sigma(\ch{N})$ construct recursively the tree for $F_{\bar{\sigma}}$ and attach it to the root node.
Because the size of the $\varepsilon$-net is independent of the size of the polyhedron, we obtain a hierarchical structure being a tree rooted at~$\ch{N}$ with maximum degree $O(|N|^{\lfloor d/2\rfloor})$.

\begin{lemma}\label{lemma:Hierarchical Construction in R^d}
Given a polyhedron $P$ in $\mathbb{R}^d$ with combinatorial complexity $n$ and any $\delta>0$, we can compute a hierarchical tree for $P$ with $O(\log  n)$ depth in $O(n^{\lfloor d/2\rfloor+\delta})$ time using $O(n^{\lfloor d/2\rfloor+\delta})$~space. 
\end{lemma}
\begin{proof}
Because we reduce the number of faces of the original polyhedron by a factor of~$\varepsilon$ on each branching of the hierarchical tree, the depth of this tree is $O(\log n)$. 

The space $S(n)$ of this hierarchical tree of $P$ can be described by the following recurrence 
$S(n) = O(|N|^{\lfloor d/2\rfloor}) S(\varepsilon n) + O(1).$
Recall that $|N| = O(\frac{1}{\varepsilon} \log \frac{1}{\varepsilon})$. Moreover, if we let $r = 1/\varepsilon$, we can solve this recurrence using the master theorem and obtain that $S(n) =  O(n^{\frac{\lfloor d/2\rfloor \log (C  r \log r)}{\log r}})$ for some constant $C>0$. Therefore, by choosing $r = 1/\varepsilon$ sufficiently large, we obtain that the total space is $S(n) = O(n^{\lfloor d/2\rfloor + \delta})$ for any $\delta>0$ arbitrarily small. 
To analyze the time needed to construct this hierarchical tree, recall that an $\varepsilon$-net can be computed in linear time~\cite{ConstructionEpsilonNets} which leads to the following recurrence $T(n) = O(|N|^{\lfloor d/2\rfloor}) S(\varepsilon n) + O(n)$. Using the same arguments as for the space we solve this recurrence and obtain that the total time is $T(n) = O(n^{\lfloor d/2\rfloor + \delta})$ for any $\delta>0$ arbitrarily small. 
\end{proof}

\subsection*{Testing intersection in higher dimensions}

Using hierarchical trees, we extend the ideas used for the 3D algorithm presented in Section~\ref{section:3D algorithm} to higher dimensions.
We start by describing the preprocessing of a polyhedron.

\subsection*{Preprocessing}
Let $P$ be a polyhedron $\mathbb{R}^d$ with combinatorial complexity $n$. Assume without loss of generality that the origin lies in the interior of $P$.
Otherwise, modify the coordinate system. 
To preprocess $P$, we first compute the polyhedron $\polarzero{P}$ being the polarization of $P$. 
Then, we compute two hierarchical trees as described in the previous section, one for $P$ and another for $\polarzero{P}$.
Similarly to the 3D case, because only one of the two hierarchical trees will ever be used in a single intersection query, we can independently use Simulation of Simplicity~\cite{edelsbrunner1990simulation} in the construction of each of the trees.
Because $|F(\polarzero{P})| =|F(P)|  = n$ by Corollary 2.14 of~\cite{ziegler1995lectures}, the total size of these hierarchical trees is $O(n^{\lfloor d/2\rfloor + \delta})$.

\subsection*{Preliminaries of the algorithm}
Let $P$ and $R$ be two independently preprocessed polyhedra in $\mathbb{R}^d$ with combinatorial complexities $n$ and $m$, respectively.
Throughout this algorithm, we fix the coordinate system used in the preprocessing of $R$, i.e., we assume that $\zero\in R$.
For ease of notation, let $Q = \polarzero{R}$. Because $\zero\in R$, Lemma~\ref{lemma:Polar of the polar} implies that $R = \polarzero{Q}$.
Assume that $P$ and $\polarzero{Q}$ lie in a \emph{primal space} while $\polarinf{P}$ and $Q$ lie in a \emph{polar space}. 
As in the 3D-algorithm, we look at the primal and polar spaces independently and switch between them whenever necessary.

To test the intersection of $P$ and $R = \polarzero{Q}$, we use the hierarchical trees of $P$ and $Q$ computed during the preprocessing step. 
The idea is to walk down these trees using paths going from the root to a leaf while maintaining some invariants.

Throughout the algorithm, we prune the faces of $P$ and keep only those that can define an intersection. 
Formally, we consider a set $F^*(P) \subseteq F(P)$  such that $P$ intersects $\polarzero{Q}$ if and only if a face of $F^*(P)$ intersects $\polarzero{Q}$. In the same way, we prune $F(Q)$ and maintain a set $F^*(Q)\subseteq F(Q)$ such that $Q$ intersects $\polarinf{P}$ if and only if a face of $F^*(Q)$ intersects $\polarinf{P}$. If these properties hold, we say that the \emph{correctness invariant} is maintained.

At the beginning of the algorithm let $F^*(P) = F(P)$ and $F^*(Q) = F(Q)$.
In each round of the algorithm we discard a constant fraction of the vertices of either $F^*(P)$ or $F^*(Q)$ while maintaining the correctness invariant. Note that these sets are not explicitly maintained.

Throughout, we consider constant size polyhedra $P_N\subseteq P$ and $Q_N\subseteq Q$ being the convex hull of $\varepsilon$-nets of $F^*(P)$ and $F^*(Q)$, respectively.
The algorithm tests if $P_N$ and $\polarzero{Q_N}$ intersect to determine either the separation or the inverse separation invariant, both analogous to those used by the 3D-algorithm. 
Formally, the \emph{separation invariant} states that we have a hyperplane $\pi$ that separates $P_N$ from $\polarzero{Q_N}$ such that $\pi$ is tangent to $P_N$ at one of its vertices.
The \emph{inverse separation invariant} states that there is a hyperplane $\mu$ that separates $\polarinf{P_N}$ from $Q_N$ such that $\mu$ is tangent to~$Q_N$ at one of its vertices. By Theorem~\ref{theorem: Polarity of polyhedra} at least one of the invariants must hold.

At the beginning of the algorithm, we let $P_N\subseteq P$ and $Q_N\subseteq Q$ be  the convex hulls of the $\varepsilon$-nets computed for $F(P)$ and $F(Q)$ at the root of their respective hierarchical trees. Because they have constant complexity, we can test if the separation or the inverse separation invariant holds.
Depending on which invariant is established, we step into the primal or the polar space as follows (if both invariants hold, we choose arbitrarily).

\subsection*{Separation invariant.} 

If the separation invariant holds, then we have a hyperplane $\pi$ tangent to $P_N$ at a vertex $v$ such that $\pi$ separates $P_N$ from $\polarzero{Q_N}$. 
Therefore, by Lemma~\ref{lemma:Simplex lemma} there is a simplex $\sigma\in \Sigma(P_N)$ such that $\pi$ is also tangent to $\sigma$ at $v$.
Because we stored $\Sigma(P_N)$ in the hierarchical tree, we go through the $O(1)$ shell-simplices of $\Sigma(P_N)$ to find $\sigma$. 
Recall that $F_{\bar{\sigma}}$ is the set of faces of $F^*(P)$ that intersect the complement of~$\sigma$. 
Thus, every face of $P$ intersecting the halfspace $\planezero{\pi}$ belongs to~$F_{\bar{\sigma}}$.

Because $\pi$ separates $P_N$ from $\polarzero{Q}\subseteq \polarzero{Q_N}\subseteq \planezero{\pi}$, the only faces of $F^*(P)$ that could define an intersection with~$\polarzero{Q}$ are those in $F_{\bar{\sigma}}$, i.e., a face of $F^*(P)$ intersects $\polarzero{Q}$ if and only if a face of $F_{\bar{\sigma}}$ intersects $\polarzero{Q}$.
Because the correctness invariant held prior to this step, we conclude that a face of $F_{\bar{\sigma}}$ intersects $\polarzero{Q}$ if and only if $P$ intersects $\polarzero{Q}$.

Recall that we have recursively constructed a tree for $F_{\bar{\sigma}}$ which hangs from the node storing $P_N$.
In particular, in the root of this tree we have stored the convex hull of an $\varepsilon$-net of $F_{\bar{\sigma}}$.
Therefore, after finding $\sigma$ in $O(1)$ time, we move down one level to the root of the tree of $F_{\bar{\sigma}}$. Then, we redefine $P_N$ to be the convex hull of the $\varepsilon$-net of $F_{\bar{\sigma}}$ stored in this node. 
Moreover, we let $F^*(P) = F_{\bar{\sigma}}$ which preserves the correctness invariant.
Then, we test if the new $P_N$ and $\polarzero{Q_N}$ intersect to determine if either the separation or inverse separation invariant holds. In this way, we moved down one level in the hierarchical tree of~$P$ and proceed with a new round of the algorithm.

\subsection*{Inverse separation invariant.} 
If the inverse separation invariant holds, then we have a hyperplane that separates $\polarinf{P_N}$ from $Q_N$. Applying an analogous procedure to the one described for the separation invariant, we redefine $Q_N$ and move down one level in the hierarchical tree of $Q$ while maintaining the correctness invariant. Then, we test if $\polarinf{P_N}$ intersects the new $Q_N$ to determine if either the separation or inverse separation invariants holds and proceed with the algorithm.
\vspace{.2in}

After $O(\log n + \log m)$ rounds, the algorithm reaches the bottom of the hierarchical tree of either $P$ or $Q$.
If we reach the bottom of the hierarchical tree of $P$ and the separation invariant holds, then because $\polarzero{Q_N}\supseteq \polarzero{Q}$ by Lemma~\ref{lemma:Inverse of inclusion},
 we have a hyperplane that separates $P_N = \ch{F^*(P)}$ from $\polarzero{Q_N}$. 
 That is, no face of $F^*(P)$ intersects $\polarzero{Q_N}$.
 Because $P$ and $\polarzero{Q}$ intersect if and only if a face of $F^*(P)$ intersects $\polarzero{Q}$ by the correctness invariant, we conclude that $P$ and $R = \polarzero{Q}$ do not intersect.

Analogously, if we reach the bottom of the hierarchical tree of $Q$ and the inverse separation invariant holds, then we have a hyperplane that separates $Q_N = \ch{F^*(Q)}$ from $\polarinf{P_N}\supseteq \polarinf{P}$. 
That is, no face of $F^*(Q)$ intersects $\polarinf{P}$.
Thus, by the correctness invariant, we conclude that $Q$ and $\polarinf{P}$ do not intersect. Therefore, Theorem~\ref{theorem: Polarity of polyhedra} implies that $P$ and $R = \polarzero{Q}$ intersect.

In any other situation the algorithm can continue until one of the two previously mentioned cases arises and the algorithm finishes. 
Recall that the hierarchical trees of $P$ and $Q$ have logarithmic depth by Lemma~\ref{lemma:Hierarchical Construction in R^d}.
Because in each round we move down in the hierarchical tree of either $P$ or $Q$, after $O(\log n + \log m)$ rounds the algorithm finishes. Moreover, since each round can be performed in $O(1)$ time, we obtain the following result.

\begin{theorem}
Let $P$ and $R$ be two independently preprocessed polyhedra in $\mathbb{R}^d$ with combinatorial complexities $n$ and $m$, respectively. For any given translations and rotations of $P$ and~$R$, we can determine if $P$ and $R$ intersect in $O(\log n + \log m)$ time.
\end{theorem}

Note that this algorithm does not construct a hyperplane that separates $P$
and $\polarzero{Q}$ or a common point, but only determines if
such a separating plane or intersection point exists.
In fact, if $P$ is disjoint from $\polarzero{Q}$, then we can take the $O(\log n)$ hyperplanes found by the algorithm, each of them separating some portion of $P$ from $\polarzero{Q}$. Because all these hyperplanes support a halfspace that contains $\polarzero{Q}$, their intersection defines a polyhedron $S$ that contains $\polarzero{Q}$ and excludes $P$. Therefore, we have a certificate of size $O(\log n)$ that guarantees that $P$ and $\polarzero{Q}$ are separated.

Similarly, if $Q$ is disjoint from $\polarinf{P}$, then we can find a polyhedron of size $O(\log m)$ whose boundary separates $Q$ from $\polarinf{P}$. In this case, we have a certificate that guarantees that $Q$ and $\polarinf{P}$ are disjoint which by Theorem~\ref{theorem: Polarity of polyhedra} implies that $P$ and $\polarzero{Q}$ intersect. 

\subsection*{Acknowledgments.} We thank David Kirkpatrick and anonymous referees for useful comments in a previous version of this paper.

{\small
\bibliographystyle{abbrv}
\bibliography{PolyhedraDistance}
}

\end{document}